\documentclass[letterpaper,11pt]{amsart}

\usepackage{amsaddr3}
\usepackage{graphicx}
\usepackage{amsmath}
\usepackage{mathrsfs}
\usepackage[numbers,square]{natbib}
\usepackage{verbatim}

\numberwithin{equation}{section}
\numberwithin{figure}{section}
\numberwithin{table}{section}

\usepackage{color}
\newcommand{\set}[2]{#1_{(#2)}}

\newcommand{\intii}{\int_{-\infty}^\infty}

\newcommand{\vp}{\varphi}
\newcommand{\ud}{\, \mathrm{d}}
\newcommand{\abs}[1]{\left \vert #1 \right \vert}
\newcommand{\s}{\sigma}
\renewcommand{\a}{\alpha}
\renewcommand{\b}{\beta}
\renewcommand{\t}{\theta}
\newcommand{\g}{\gamma}
\renewcommand{\d}{\Delta}
\renewcommand{\l}{\lambda}
\newcommand{\inv}[1]{{#1}^{-1}}
\newcommand{\iid}{\stackrel {\textrm{iid}}{\sim}}

\newcommand{\cov}{\textrm{cov}}

\newcommand{\rv}[3][1]{#2_{#1},\ldots,#2_{#3}}
\newcommand{\N}{\mathcal N}
\newcommand{\F}{\mathcal F}
\renewcommand{\S}{\mathcal S}

\renewcommand{\|}{\ | \ }
\newcommand{\comp}{\textrm{comp}}
\newcommand{\miss}{\textrm{miss}}
\newcommand{\obs}{\textrm{obs}}
\newcommand{\new}{\textrm{new}}
\newcommand{\old}{\textrm{old}}
\newcommand{\del}[2][]{\frac{\partial^{#1}}{\partial {#2}^{#1}}}
\newcommand{\fdel}[3][]{\frac{\partial^{#1}#3}{\partial{#2}^{#1}}}

\newcommand{\grad}{\nabla}

\newcommand{\diag}{\textrm{diag}}
\DeclareMathAlphabet{\mathpzc}{OT1}{pzc}{m}{it}
\newcommand{\f}{\mathpzc f}
\newtheorem{proposition}{Proposition}

\newtheorem{theorem}{Theorem}

\newcommand{\be}{\begin{align}}
\newcommand{\ee}{\end{align}}

\begin{document}

\title[SDEs with Memory Effects: An Inference Framework]{Statistical Inference for Stochastic Differential Equations with Memory}

\author{Martin Lysy$^1$}
\address{$^1$Department of Statistics and Actuarial Science\\ University of Waterloo}
\email{mlysy@uwaterloo.ca}
\date{July 2, 2013}

\author{Natesh S. Pillai$^2$}
\address{$^2$Department of Statistics\\ Harvard University}
\email{pillai@stat.harvard.edu}

\begin{abstract}
In this paper we construct a framework for doing statistical inference for discretely observed stochastic differential equations (SDEs) where the driving noise has `memory'. Classical SDE models for inference assume the driving noise to be Brownian motion, or ``white noise'', thus implying a Markov assumption.  We focus on the case when the driving noise is a fractional Brownian motion, which is a common continuous-time modeling device for capturing long-range memory.  Since the likelihood is intractable, we proceed via data augmentation, adapting a familiar discretization and missing data approach developed for the white noise case.  In addition to the other SDE parameters, we take the Hurst index to be unknown and estimate it from the data.  Posterior sampling is performed via a Hybrid Monte Carlo algorithm on both the parameters and the missing data simultaneously so as to improve mixing.  We point out that, due to the long-range correlations of the driving noise, careful discretization of the underlying SDE is necessary for valid inference. Our approach can be adapted to other types of rough-path driving processes such as Gaussian ``colored'' noise. The methodology is used to estimate the evolution of the memory parameter in US short-term interest rates. 
\end{abstract}

\maketitle



\section{Introduction}
In this paper we develop a framework based on data augmentation for performing statistical inference for discretely observed stochastic differential equations (SDEs) driven by non-Markovian noise such as fractional Brownian motion.  SDEs are routinely used to model continuous-time phenomena in the natural sciences~\citep{ashby-lieber04, golightly-wilkinson05, gillespie07}, engineering~\citep{pardoux-pignol84, whitmore95, sobczyk01}, and finance~\citep{cox-et-al85, heston93, hull-white01}.  
Consider an SDE with drift $\mu$ and diffusion coefficient $\sigma$ denoted as
\begin{equation}\label{sdewhite}
\ud X_t = \mu(X_t, \theta) \ud t + \s(X_t, \theta) \ud B_t, \quad X_0 \in \mathbb{R},
\end{equation}
where $B_t$ is a standard one-dimensional Brownian motion and $\theta$ is a parameter of interest.  The stochastic process $X_t$ is commonly referred to as a \emph{diffusion process}.  It is a strong Markov process with continuous sample paths~\citep{karlin-taylor81}.  Remarkably, almost all stochastic processes with these two properties satisfy an SDE of the form~\eqref{sdewhite}~\citep[see][for discussion and counter-example]{vankampen82}.

An equation such as~\eqref{sdewhite} specifies the stochastic evolution of $X_t$ on an infinitesimal time scale.  That is, suppose that $X_0$ is given and we wish to simulate the path of $X_t$ on the interval $[0, T]$.  For $\d t = T/N$, setting $\hat X_0 = X_0$, the usual Euler (or Euler-Maruyama) scheme is 
\begin{equation}\label{emapprox}
\hat X_{(n+1)\d t} = \hat X_{n\d t} + \mu(\hat X_{n \d t},\t) \d t + \sigma(\hat X_{n \d t},\t) \d B_n,
\end{equation}
where $\d B_n = (B_{(n+1)\d t} - B_{n \d t}) \iid \N(0, \d t)$. 
The continuous process $\hat X_t^{(\d t)}$ obtained by interpolation converges to $X_t$ as $\d t \to 0$ in an appropriate sense~\citep{maruyama55}.

The above discrete-time approximation provides a fundamental intuition for modeling physical phenomena using continuous-time stochastic processes: $\mu(X_t) \d t$ is the infinitesimal change in mean and $\sigma(X_t) \d t$ is the infinitesimal variance. Most of the existing statistical inference methodology for discretely observed diffusions crucially utilize such discretization schemes:~\citep{pedersen95, liu-sabatti00, eraker01, durham-gallant02, golightly-wilkinson05, chib-et-al10, kou-et-al12} all do so directly;~\citep{roberts-stramer01, golightly-wilkinson08, kalogeropoulos-et-al10, beskos-et-al12} use it indirectly to evaluate the Girsanov change-of-measure.


While the Markov assumption for the observed data -- central to diffusion modeling -- is justifiable in many situations, there is a growing number of applications in which it is not. For instance, the dynamics of financial data~\citep{cheridito03}, subdiffusive proteins~\citep{kou-xie04}, and internet traffic and networks~\citep{stoev-et-al05,wolpert-taqqu05} all exhibit spurious trends and fluctuations which persist over long periods of time.  Such long-range dependence -- or memory -- typically leads to Markov models which are overparametrized, in order to compensate for their rapid decorrelation.

\subsection{SDEs Driven by Fractional Brownian Motion}

In an SDE such as~\eqref{sdewhite}, the Brownian motion $B_t$ can be thought of as the force which drives $X_t$.  While $B_t$ is not differentiable in the usual sense, its derivative $b_t$ (defined using the Fourier transform) can be identified with a collection of iid Normals,
such that
\[
\cov(b_{s+t}, b_s) = \delta(t).
\]
The derivative of Brownian motion is often referred to as ``white noise'', the term being derived from its flat frequency spectrum:
\[
\S(\f) = \intii e^{-2\pi i t \f} \delta(t) \ud t = 1,
\]
the spectrum of white light.  In this paper, we wish to study the solutions of SDEs which are driven by different types of noise:
\begin{equation}\label{sdecol}
\ud X_t = \mu(X_t,\theta) \ud t + \s(X_t,\theta) \ud G_t,
\end{equation}
where $G_t$ is not Brownian motion but rather a non-Markovian process. Whenever it exists, the derivative of $G_t$ is referred to as \emph{colored noise}, by a similar identification of its frequency spectrum with the colors of light~\citep{hanggi-jung95}.

While our framework in this paper is applicable to a wide range of non-Markovian driving noise processes, we focus here on the case where $G_t$ is a fractional Brownian motion,
$G_t = B_t^H$,  with Hurst parameter $0 < H < 1$.  Fractional Brownian motion (fBM) is a continuous mean-zero Gaussian process with covariance
\[
\cov(B_t^H, B_s^H) = \tfrac 1 2 \left(\abs{t}^{2H} + \abs{s}^{2H} - \abs{t-s}^{2H}\right).
\]
The Hurst, or \emph{memory} parameter $H$ indexes the self-similarity of $B_t^H$.  For any $c > 0$, we have
\[
B^H_{ct} \stackrel{d}{=} c^{H} B_t^H.
\]
While fBM itself is non-stationary, its increments
\[
\d B^H_{n} = B^H_{(n+1) \d t} - B^H_{n \d t}
\] 
form a stationary Gaussian process with autocorrelation
\begin{equation}\label{fbmacf}
\cov(\d B_n^H, \d B_{n+k}^H) = \tfrac 1 2 (\d t)^{2H} \left(\abs {k+1}^{2H} + \abs{k-1}^{2H} - 2\abs{k}^{2H} \right).
\end{equation}
For $H = \tfrac 1 2$, the fBM increments are uncorrelated, and $B_t^H = B_t$ reduces to the standard Brownian motion.  For $H \neq \tfrac 1 2$, the increments exhibit a power law decay, in contrast to the exponential decorrelation of stochastic processes with short-range memory.  The increments are positively correlated for $H > \tfrac 1 2$, and negatively for $H < \tfrac 1 2$.  
Being the only continuous, self-similary Gaussian process with stationary increments~\citep{embrechts-maejima02}, fBM occupies a central role in the history of long-range dependence modeling~\citep{mandelbrot-wallis68, cox84, samorodnitsky06}.  SDE such as~\eqref{sdecol} driven by fractional Brownian motion give rise to long-memory processes which need not be Gaussian, while harnessing the statistical power of infinitesimal-time models.


From a modeling perspective, a natural interpretation of the stochastic process defined by~\eqref{sdecol} is as the limit of a discrete-time approximation.  The Euler scheme for discretizing ~\eqref{sdecol} reads
\begin{equation}\label{emcol}
\hat X_{(n+1) \d t} = \hat X_{n \d t} + \mu(\hat X_{n \d t}, \theta) \d t + \s(\hat X_{n \d t}, \theta) \,\d B^H_n.
\end{equation}
  %
 This interpretation at first seems very promising but it turns out that, due to the ``roughness'' of the sample paths of the fractional Brownian motion (which are almost surely not differentiable for any $0 < H < 1$), the Euler scheme in \eqref{emcol} need not converge as the discretization time step $\d t \rightarrow 0$.  For instance, consider the SDE
  \[
 dX_t = X_t \ud B_t^H
 \]
 with initial value $X_0 = 1$.  The exact solution is $X_t  = \exp{(B^H_t)}$, whereas the solution of the Euler scheme at time $t=1$ is given by \cite{neuenkirch-et-al10}
 \[
 \hat X_1 =  \prod_{k =0}^{N-1} (1 + \d B^H_k).
 \]
 Using the above, for sufficiently large $N=1/\d t$, it can be shown that \cite{neuenkirch-et-al10}
 \[
 X_1 - \hat X_1 = \exp(B_1) - \exp\left(B_1 - \frac{1}{2} \sum_{k=0}^{N-1} \abs{\d B_{k}}^2 + \rho_N\right),
 \]
 where $\rho_N \rightarrow 0$ almost surely as $\d t \to 0$, for $H > \tfrac 1 3$.  However, we also know that
 for $H < \tfrac 1 2$,
 \[
 \sum_{k=0}^{N-1} |\d B_{k}|^2 \rightarrow \infty
 \]
 almost surely as $\d t \to 0$. 
Consequently, $\hat X_1$ converges to $0$ almost surely, such that the Euler-Maruyama scheme fails for $H < \tfrac 1 2$.  
On the other hand, for $ H > \tfrac 1 2$, the Euler-Maruyama scheme does converge (see Proposition \ref{eqn:EMsch} below).  Thus, the physical intuition provided by Euler-Maruyama discretization for diffusions need to be refined in the case of SDEs driven by fBM.  
However, a lot of the ideas to follow do carry over from diffusions -- all that is essential in our framework is a numerical scheme which correctly approximates the underlying SDE.

\subsection{Review of Previous Work}

Unlike diffusions, parameter estimation for SDEs driven by fractional Brownian motion is in its infancy.  There are two key challenges: the likelihood is intractable and the data is not Markovian. Some earlier works for parameter estimation include \cite{kleptsyna-et-al00, kleptsyna-lebreton02, prakasarao04, prakasarao05, lebreton-roubaud00, hult03}.  A wealth of information is contained in the book \cite{prakasarao11}. However most of these works deal with continuous data. A few papers study parameter estimation for discretely observed fractional Ornstein-Uhlenbeck processes~\cite{xiao-et-al11, prakasarao11}.

A pioneering work dealing with discrete observations is \cite{tudor-viens07} in which the authors consider a SDE of the form 
\[
X_t = \theta \int_0^t  b(X_s) ds + B_t^H
\]
where $b$ is a known function. It is known that fBM can be represented as an It\=o integral,
\begin{equation} \label{eqn:weinrep}
B_t^H = \int_0^t K_H(t,s)\ud W_s,
\end{equation}
where $W_t$ is a standard Brownian motion and $K_H$ is a kernel (see \cite{tudor-viens07} for details).  This representation leads to a version of Girsanov's theorem \cite{norros99} which can be used for computing the likelihood function.  For continuously observed $X_t$, using this version of Girsanov's theorem, the authors derive the maximum likelihood estimator $\hat{\theta}_\mathrm{con}$ for $\theta$ and show that it is consistent for all $H \in (0,1)$. For discrete data, since the MLE is hard to derive, the authors study a discretized appoximation of $\hat{\theta}_{\mathrm{con}}$ and prove its consistency. 
 
More recently, there have been a couple of different approaches for discrete data which avoid a direct likelihood computation. In \cite{neuenkirch-tindel11}, the authors construct a least squares-type procedure for parameter estimation in SDEs driven by fBM with constant diffusion coefficient (but assume $H > \tfrac 1 2$) and show its consistency.  In \cite{saussereau11}, the author considers a SDE of the form
 \[
  X_t  = x_0 + \int_0^t b(X_s) ds + \sigma B_t^H
 \]
 and constructs a nonparametric kernel estimator for the drift coefficient $b$. In \cite{kaur-et-al11}, the authors construct estimating functions for $\t$ in SDEs with linear drift, $b(x) = a(x) + \t c(x)$.  Finally in \cite{chronopoulou-tindel13}, the authors construct an interesting maximum likelihood type estimator using tools from Malliavin calculus. Curiously, \cite{chronopoulou-tindel13} does not take the non-Markovinanity of the data into account while computing the quasi-likelihood function.  Most of the above papers only deal with point estimation and do not give uncertainty quantification. We also note that almost all of them assume $H$ to be known, and thus do not venture into estimating $H$.

\subsection{An Inference Framework based on Data Augmentation}\label{data-aug}

Due to the non-Markovianity, the likelihood function for discrete data is intractable. We proceed via a data augmentation approach, by ``filling in'' data between two observed points so as to better approximate the likelihood function. There are many equivalent approaches for defining and constructing approximations for the solutions of  non-white noise SDE~\eqref{sdecol} involving mathematical machinery such as Malliavin calculus~\citep{tudor-viens07, hu-yan09}, Wick products~\citep{oksendal09, hu-yan09}, and generalized distributions~\citep{zindewalsh-phillips03}.  Compared to these approaches, data augmentation based on the Euler scheme~\eqref{emcol} -- should it apply -- is conceptually simpler and leads directly to a longstanding inference framework developed for the white noise case~\citep{pedersen95, eraker01, kou-et-al12}.

That is, consider the SDE with constant diffusion
\[
\ud X_t = \mu(X_t, \t)\ud t + \s \ud G_t,
\]
where $G_t$ is any continuous stochastic process.  Let $X = X_\obs = (\rv [0] X N)$ be discrete observations of the SDE at regular time intervals $\d T$.  Both the drift and diffusion functions depend on parameters $\t$ and $\s$ which are to be estimated from the data.

For given level $k$, define the complete data
\[
X_{(k)} = X_\comp = (\rv [k,0] X {k,M_k}),
\]
where $X_{k,n}$ corresponds to an observation of $X_t$ at time $t = n \d t_k$, with $\d t_k = \d T/2^k$.  Thus, $X_{(0)} = X_\obs$, and in general, $X_n = X_{k,n2^k}$.

To the extent that the discrete time appproximation~\eqref{emcol} is correct, the approximate complete data likelihood is
\begin{equation}\label{colll}
\log(\hat L(\t,\s \| X_{(k)})) = \log(f(\d G_{(k)})) - M_k\log(\s),
\end{equation}
where $f(\d G_{(k)})$ is the density of the noise increments,
\begin{align*}
\d G_{(k)} & = (\d G_{k,0}, \ldots, \d G_{k, M_k-1}), \\
\d G_{k,n} & = \tfrac 1 \s \left\{\d X_{k,n} - \mu(X_{k,n},\t) \d t_k\right\},
\end{align*}
with $\d X_{k,n} = X_{k,n+1}-X_{k,n}$.  In the white noise case, the $\d G_{k,n}$ are iid Normals.  In this study, we take $\d G_{(k)} = \d B_{(k)}^H$ to be fBM increments with density $f(\d B_{(k)}^H \| H)$, corresponding to a mean-zero stationary Gaussian process with covariance function given by~\eqref{fbmacf}.

With the approximate likelihood $\hat L(\t, \s, H \| \set X k)$, Bayesian inference can be realized by specifying a prior $\pi(\t, \s, H)$, and sampling from the joint distribution of the parameters $(\t,\s,H)$ and the missing data $X_\miss = X_\comp \setminus X_\obs$:
\begin{equation}\label{comppost}
\hat p_k(X_\miss, \t, H \| X_\obs) \propto \hat L(\t, H \| X_{(k)}) \pi(\t, H).
\end{equation}
Such a strategy is appropriate when the approximate posterior distribution
\[
\hat p_k(\t, \s, H \| X_\obs) = \int \hat p_k(X_\miss, \t, \s, H \| X_\obs) \ud X_\miss
\]
converges to the true SDE posterior $p(\t, \s, H \| X_\obs)$ as $k \to \infty$.  This assumption is repeatedly employed in the white noise literature~\citep{eraker01, golightly-wilkinson08}, and generally seems to hold in practice (despite some theoretical results that would suggest the contrary~\citep{cano-et-al06}).

\subsection{Outline of the Paper}
The remainder of this article is organized as follows.  In Section~\ref{rough-paths}, we use  the Doss-Sussman approach to define a solution of the SDE~\eqref{sdecol} for any continuous process $G_t$.  In the white noise case, this is equivalent to the Stratonovich interpretation of the SDE, which is identical to the more familiar It\=o interpretation~\eqref{emapprox} when the diffusion $\s(x,\t)  \equiv \s$ is constant.  This is precisely the case when approximate inference by way of the natural Euler scheme~\eqref{emcol} for fBM-driven SDEs converges to the result of the true posterior as $k\to \infty$.  Furthermore, a change-of-variables is presented which reduces most SDEs of interest to the constant diffusion case.

In Section~\ref{sampling-alg}, a Markov Chain Monte Carlo (MCMC) algorithm is presented for sampling from the approximate posteriors.  In order to address well-known mixing-time issues carrying over from the white noise case~\citep{roberts-stramer01, kou-et-al12, beskos-et-al12}, we employ a Hybrid Monte Carlo (HMC) sampling strategy.

Sections~\ref{fOUex} and~\ref{fCIRex} illustrate the methodology with two different models: the fractional Ornstein-Uhlenbeck (fOU) process and the fractional Cox-Ingersoll-Ross (fCIR) process.  The latter is to our knowledge defined here for the first time when $H < \tfrac 1 2$.  We use these models to investigate the evolution of the memory parameter in short-term interest rates on US Treasury Bills, from January 1954 to June 2013.  Interestingly, both models lead to very similar conclusions -- significant evidence of positive long-term correlation -- right up to the Global Financial Crisis of 2007-2008.

\section{An SDE Discretization Scheme by Rough-Paths}\label{rough-paths}

For SDEs driven by fractional Brownian motion, the integral form of \eqref{sdecol} becomes
\begin{equation} \label{eqn:integform}
X_t = x_0 + \int_0^t \mu(X_s, \theta) \ud s + \int_0^t \sigma(X_s, \theta) \ud B_s^H.
\end{equation}
A key difficulty for inference with \eqref{eqn:integform} when $\s(X_t,\t)$ is a non-constant function of $X_t$ is to make sense of the integral $\int_0^t \s(X_s) \ud B_s^H$ for various integrands $\s : \mathbb{R} \mapsto \mathbb{R}$. 

We define the solution of the stochastic integral  $\int_0^t \s(X_s) dB_s^H$ using the classical Doss-Sussman transformation \cite{sussmann78,nourdin-simon06}.  The idea behind this transformation is that, since $B^H_t$ has continuous sample paths, the solution of the stochastic equation for each sample path $\{B^H_s, 0 \leq s \leq t \}$ can be obtained by solving a corresponding ordinary differential equation. Thus we have a \emph{pathwise} solution for the stochastic integral, in contrast to the classical It\=o integral for diffusions where the solution is defined as an $L_2$ limit of partial sums.  Thus, the developments presented here are valid not only for integrals involving fBM, but for any stochastic process $G_t$ with continuous paths replacing $B_t^H$ in~\eqref{eqn:integform}.

Indeed, let $f,g$ be functions such that (i) $g$ is continuously differentiable and (ii) $f$ and $g'$ are locally Lipschitz.  Then for the SDE 
\begin{equation} \label{eqn:ysde}
Y_t = y_0 + \int_0^t f(Y_s) \ud s + \int_0^t g(Y_s) \ud B_s^H
\end{equation} 
the Doss-Sussman transformation yields the solution as
\begin{equation} \label{eqn:dssol}
Y_t = \varphi(B_t^H, Z_t),
\end{equation}
where the function $\varphi(x,y) : \mathbb{R}^2 \mapsto \mathbb{R}$ satisfies 
$\tfrac{\partial}{\partial x} \varphi(x,y) = g(\varphi(x,y))$, $\varphi(0,y) = y$ for all $y \in \mathbb{R}$, and the process $Z_t$ solves the random ordinary differential equation
\begin{equation} \label{eqn:adiff}
Z_t = y_0 + \int_0^t a(B^H_s, Z_s) \ud s,
\end{equation}
where 
\begin{equation}
a(x,y) = f(\varphi(x,y)) \exp\left\{ - \int_0^x g'(\varphi(u,y)) \ud u \right\}.
\end{equation}
Under the conditions (i) and (ii) above on $f$ and $g$, the solution \eqref{eqn:dssol} is unique \cite{sussmann78}.

A few remarks are in order. First, the Doss-Sussman solution is valid for any $H \in (0,1)$ due to the continuity of the sample paths of $B^H_t$.  In fact, under conditions (i) and (ii) the solution paths $Y_t$ are continuous themselves.  Second, for $H = \tfrac 1 2$, the solution given by the Doss-Sussman transformation is the same as the classical Stratonovich integral~\cite{sussmann78},
\[
Y_t = Y_0 + \int_0^t f(Y_s) \ud s + \int_0^t g(Y_s) \circ \ud B_t.
\]
Third, for the Doss-Sussman notion of the solution, we have the following change-of-variables:
\begin{equation}\label{ds-cov}
h(Y_t) = h(Y_0) + \int_0^t h'(Y_s)f(Y_s) \ud s + \int_0^t h'(Y_s) g(Y_s) \ud B^H_s.
\end{equation}
This can easily be checked when $h$ is invertible by constructing the solution to the SDE with drift $h'(h^{-1}(x))f(h^{-1}(x))$ and diffusion $h'(h^{-1}(x))g(h^{-1}(x))$.  This result is consistent with the well-known change-of-variables formula for Stratonovich integrals when $H = \tfrac 1 2$, and also with the formula of~\cite{dai-heyde96} for $H > \tfrac 1 2$.  Crucially for what follows, this allows us to reduce many SDEs of interest to having constant diffusion by taking $h(y) = \int 1/g(y) \ud y$.  
Finally, both for statistical intuition and for the inference methodology to follow, we shall require a numerical approximation to the solution of~\eqref{eqn:ysde}.
 
   For a function $f: [0,T] \mapsto \mathbb{R}$,
define the $\alpha$-H\"older norm to be
\begin{equation}
\abs{f}_\alpha = \sup_{s,t \in [0,T], s\neq t} \frac{\abs{f(t) - f(s)}}{\abs{t-s}^\alpha}.
\end{equation}
When $\abs{f}_\alpha < \infty$, we will denote $f \in C^\alpha$.
The sample paths of $B_t^H$ are H\"older continuous for any $\alpha < H$ almost surely.  Thus the sample paths of the fBM become more regular (or less``rough'') as $H$ increases towards $1$. In fact, the stochastic integration theory for $H > \tfrac 1 2$ is different from the theory for $H < \tfrac 1 2$. 

\subsection{Integrals with $H > \tfrac 1 2$}  For functions $f,g \in C^\gamma$ with $\gamma > \tfrac 1 2$, the Riemann sum $ \sum_\pi f_i (g_{i+1} - g_i)$ converges as the partition size $|\pi| \rightarrow 0$ \cite{young36}.  Thus for such $f,g$ we may define the Riemann-Stieljes integral
\begin{equation} \label{eqn:intyou}
\int f \ud g = \lim_{\abs{\pi} \rightarrow 0} \sum_\pi f_i (g_{i+1} - g_i).
\end{equation}
The above integral is called the Young integral \cite{young36}.  The solution using the Doss-Sussman transformation for stochastic integrals $\int_0^t g(X_s) \ud B_s^H$ for $H > \tfrac 1 2$ coincides with that of the Young integrals \cite{nourdin08}.

The advantage of the Young integral is that it gives a way to discretize the stochastic integral for numerical simulation. This is also central to our inference strategy.
In the SDE~\eqref{eqn:ysde}, let $f,g$ be smooth and bounded functions, and consider the simple Euler-Maruyama discretization scheme~\eqref{emcol}. For $H > \tfrac 1 2$, we have the following \emph{pathwise} convergence result \cite[Proposition 4.2] {chronopoulou-tindel13}(also see
 \cite{deya-et-al12, neuenkirch-nourdin07, mishura-shevchenko08} for similar results):
\begin{proposition} \label{eqn:EMsch}
Fix any $T > 0$, and let $\hat Y_t^{(\d t)}$ be a continuous-time interpolation of an Euler-Maruyama approximation as defined in~\eqref{emcol}. Then for any $\rho > 0$, there exists a random variable $C_T$ with
all $L_p$ moments such that
\begin{equation}
\sup_{t \in [0,T]} |\hat{Y}^{(\d t)}_t - Y_t| = C_T \cdot \d t^{2H - 1 - \rho},
\end{equation}
where $Y_t$ is given by \eqref{eqn:ysde}.
\end{proposition}

\subsection{Integrals with $H < \tfrac 1 2$}

For $H < \tfrac 1 2$, the stochastic integral cannot be given a pathwise definition using the Riemann-Stieljes integral.  However, a pathwise definition for stochastic integrals can still be achieved using higher order discretization schemes~\cite{neuenkirch-et-al10, deya-et-al12}. 
We do not venture into the details of this here but rather note that, already for $H = \tfrac 1 2$, the Stratanovich integral does have a pathwise definition, since it is the same as the solution obtained from the Doss-Sussman equations.   

\subsection{Stochastic Integration for Constant Diffusion}
As mentioned before, the Euler-Maruyama scheme \eqref{emcol} will generally not converge for $H < \tfrac 1 2$.  However, a notable exception is when the diffusion coefficient is constant,
\[
Y_t = Y_0 + \int_0^t \mu(Y_s) \ud s + \s B_t^H.
\]
Indeed, we have an exact discretization of the form
\[
\hat Y_{n+1} = \hat Y_n + \int_{n \d t}^{(n+1)\d t} \mu(Y_s) \ud s + \s \d B_n^H,
\]
and since $Y_s$ is continuous, the integral can be approximated in the usual
way for small $\d t$:
\[
\int_{n \d t}^{(n+1)\d t} \mu(Y_s) \ud s \approx \mu(Y_n) \d t.
\]
The following is a precise statement of this result.  For fixed $T > 0$, let $\d t = T/N$ and $\hat Y_{n \d t}$be given by the Euler scheme~\eqref{emcol}.  Let $\hat Y_t^{(\d t)}$ be a continuous-time stochastic process obtained from linear interpolation of the
Euler scheme:
\[
\hat Y_t^{(\d t)} = \tfrac{(t-n\d t)}{\d t} \hat Y_{(n+1)\d t} + \left(1-\tfrac{t-n\d t}{\d t}\right) \hat Y_{n \d t}, \quad n \d t \le t \le (n+1)\d t.
\]
\begin{theorem} \label{thm:linint} Let $\mu: \mathbb{R} \mapsto \mathbb{R}$ have a globally bounded derivative. 
Then
\[
\sup_{t \in [0,T]} \abs{\hat Y_t^{(\d t)} - Y_t}  \to 0
\]
pathwise as $\d t \to 0$.
\end{theorem}
\begin{proof}
The proof is a standard argument based on the contraction mapping theorem and thus we omit it (also see \cite{sussmann78}). 
\end{proof}

 

\section{MCMC Sampling of the Complete Data Posterior Distribution}\label{sampling-alg}

In conjunction with the change-of-variables formula~\eqref{ds-cov}, Theorem~\ref{thm:linint} extends the natural interpration of the SDE by the Euler scheme to the non-white noise case.  Furthermore, the pathwise convergence result implies that approximate inference by way of Euler-based data augmentation as described in Section~\ref{data-aug} converges to the correct result.

Suppose that~\eqref{ds-cov} allows the SDE to be transformed to constant variance,
\[
\ud X_t = \mu(X_t, \t) \ud t + \s \ud B_t^H.
\]
For observed data $X_\obs = (\rv [0] X N)$ sampling from the level-$k$ complete data posterior $p_k(X_\miss, \t, \s, H \| X_\obs)$ in~\eqref{comppost} can be a nontrivial task.  For instance, consider a Gibbs sampler which updates each component of $(X_\miss, \t, \s, H)$ individually, conditioned on all other variables.  For the white noise SDE with $H = \tfrac 1 2$, the conditional distribution of a missing data point $X_{k,n}$ depends only on its two neighbors,
\begin{align*}
p_k(X_{k,n} \| \t, \s, X_{(k)}\setminus X_{k,n}) & = p_k(X_{k,n} \| \t, \s, X_{k,n-1}, X_{k,n+1}) \\
& \propto p(X_{k,n+1} \| X_{k,n}, \t, \s) \cdot p(X_{k,n} \| X_{k,n-1}, \t, \s).
\end{align*}
While this density does not correspond to a known distribution, it can easily be evaluated, and indeed has been repeatedly used in an effective Metropolis-within-Gibbs sampling algorithm~\citep{eraker01, golightly-wilkinson05, kou-et-al12}.  In contrast, the corresponding conditional draws of $X_{k,n}$ for SDEs driven by fBM with $H \neq \tfrac 1 2$ depend on \emph{all} the complete data points $X_{k,j}$, $j \neq n$.  For high resolution $k$ this can incur a considerably higher computational cost for density evaluation.

Evaluation costs aside, it has often been pointed out in the diffusion literature that conditional draws of the parameters $p_k(\t,\s \| X_{(k)})$ and the missing data $p_k(X_\miss \| X_\obs, \t,\s)$ become increasingly correlated as $k \to \infty$~\citep{roberts-stramer01,kalogeropoulos-et-al10,kou-et-al12,beskos-et-al12}.  The upshot of this is that even an idealized Gibbs sampler which alternates between perfect draws of $X_\miss$ and $(\t,\s)$ becomes arbitrarily inefficient.  One way to overcome this problem is by adding measurement error to the model, in which case a non-centered parametrization is readily available~\citep{chib-et-al10, golightly-wilkinson08}. Within the error-free model, \cite{kou-et-al12} have implemented joint, independent proposals for $p_k(X_\miss, \t \| X_\obs)$ using parallel sampling techniques.  Within each chain, however, local MCMC updates are required.  Here, we shall consider joint updates of a different nature.

\subsection{A Basic HMC Algorithm}

Hybrid Monte Carlo (HMC) is a popular alternative to local updating strategies when Gibbs samplers and Vanilla Monte Carlo are inefficient~\citep{duane-et-al87, ishwaran99, neal10, girolami-calderhead11}, and has been successfully employed in the white noise SDE literature~\citep{beskos-et-al12}.  HMC uses Hamiltonian dynamics to construct global proposal distributions and thus does not have a diffusive behavior as that of random walk based proposals. A simple HMC algorithm taken from~\cite{liu01} is as follows.  Suppose we wish to sample a $D$-dimensional random variable $x = (\rv x D)$ with distribution $p(x) \propto \exp(-\Omega(x))$.  To do this, we first specify a mass vector $m = (\rv m D) > 0$, a small time increment $\d t$, and a step number $L$.  Then, given a previous MCMC value $x^\old$, an HMC proposal is generated and accepted by the following steps:
\begin{enumerate}
	\item Let $x_0 = x^\old$, and $p_0 \sim \N_D(0, \diag(m))$. Denote the density of this multivariate Normal distribution by $\varphi(\cdot \| m)$.
	\item For $1 \le n \le L$, calculate the deterministic recursion
	\begin{align*}
	p_{n}^{(1/2)} & = p_{n-1} - \grad \Omega(x_{n-1}) \d t/2 \\
	x_n & = x_{n-1} + m^{-1} p_{n}^{(1/2)} \d t \\
	p_n & = p_{n}^{(1/2)} - \grad \Omega(x_n) \d t/2.
	\end{align*}
	Let $x^\new = x_L$, $p^\old = p_0$ and $p^\new = p_L$.  This is the so-called ``leapfrog'' algorithm for calculating $(x, p)$~\citep{duane-et-al87}.
	\item Accept the HMC proposal $x^\new$ with Metropolis-Hastings probability
	\[
	\min\left\{1, \frac{\exp(-\Omega(x^\new)) \varphi(p^\new \| m)}{\exp(-\Omega(x^\old)) \varphi(p^\old \| m)}\right\}.
	\]
\end{enumerate}

\subsection{Efficient Derivative Evaluations for fBM Increments}

In the context of fBM-driven SDEs, the HMC algorithm above can be used to sample any subset of the posterior random variables in
\[
p_k(X_\miss, \t, \s, H \| X_\obs) \propto \exp\{-\Omega(X_{(k)}, \t, \s, H)\},
\]
where
\[
\Omega(\set X k, \t, \s, H) = -\log(f(\d B^H_{(k)} \| H)) + M_k \log(\s) - \log(\pi(\t,\s,H)).
\]
Since $f(\d B_{(k)}^H \| H)$ is a zero-mean Gaussian density, the Cholesky decomposition of its variance matrix leads to a factorization of the form
\begin{equation}\label{gaussfac}
\begin{split}
\log(f(\d B^H_{(k)} \| H)) & = -\frac 1 2 \sum_{j = 0}^{M_k-1} \frac{\left(\sum_{n = 0}^j b_{j,n} \d B^H_{n}\right)^2}{v_j} - \frac 1 2 \sum_{j=0}^{M_k-1} \log(v_j) \\
& = - \frac 1 2 \sum_{j = 0}^{M_k - 1} r_j^2/v_j + \log(v_j),
\end{split}
\end{equation}
where $b_{j,n} = b_{j,n}(H)$, $v_j = v_j(H)$, $b_{j,j} = 1$ and the subscript $k$ has been omitted to simplify notation.  Using this representation, the partial derivatives of $\Omega(\set X k, \t, \s, H)$ with respect to $X_n$, $\t_i$, and $\s$ are:
\begin{equation}\label{uprime}
\begin{split}
\fdel{X_{n}}{\Omega(X_{(k)}, \t, \s, H)} & = \sum_{j=n-1}^{M_k-1}\frac{r_j}{v_j} \left[b_{j,n} \frac{1 + \mu_x(X_n,\t)\d t_k}{\s} - \frac{b_{j,n-1}}{\s} \right] \\
\fdel{\t_i}{\Omega(X_{(k)}, \t, \s, H)} & = \frac{\pi_i(\t,\s,H)}{\pi(\t, \s, H)} + \sum_{j=0}^{M_k-1} \left[\frac{r_j}{v_j} \sum_{n=0}^j b_{j,n} \frac{\mu_i(X_n,\t)\d t_k}{\s}\right] \\
\fdel{\s}{\Omega(X_{(k)}, \t, \s, H)} & = \frac{\pi_\s(\t,\s,H)}{\pi(\t, \s, H)} - \frac{M_k}{\s} \\
& \phantom{= \ } + \sum_{j=0}^{M_k-1} \left[\frac{r_j}{v_j} \sum_{n=0}^j b_{j,n} \frac{\d X_n - \mu(X_n,\t)\d t_k}{\s^2}\right],
\end{split}
\end{equation}
with
\begin{align*}
\mu_x(x,\t) & = \del x \mu(x,\t), &  \pi_i(\t,s,H) & = \del {\t_i} \pi(\t,\s,H), \\
\mu_i(x,\t) & = \del {\t_i} \mu(x,\t), & \pi_\s(\t,\s,H) & = \del \s \pi(\t,\s,H).
\end{align*}
The advantage of using the factorization of~\eqref{gaussfac} to write the derivatives in~\eqref{uprime} is that for the stationary fBM increments, the Cholesky coefficients $b_{j,n}$ and $v_j$ can be calculated in $O(M_k^2)$ operations using the well-known Durbin-Levinson algorithm~\citep{levinson47, durbin60, brockwell-davis09}.  This is a considerable acceleration over the usual scaling of $O(M_k^3)$ for arbitrary variance matrices, and consequently for the matrix inversion required to compute the inner product in $f(\set {\d B^H} k \| H)$.

While the partial derivative $\del H \Omega(X_{(k)},\t,\s, H)$ can be evaluated analytically, the procedure is cumbersome and there doesn't seem to be a way for the Durbin-Levinson algorithm to accelerate the necessary calculations.  In the examples below, we have opted for numerical evaluation of $\del H \Omega(X_{(k)},\t,H)$ using the standard second-order difference method.


\section{Numerical Example: The Fractional Ornstein-Uhlenbeck Process}\label{fOUex}

The fractional Ornstein-Uhlenbeck (fOU) process $X_t$ satisfies the SDE
\begin{equation}\label{fou}
\ud X_t = -\gamma(X_t - \mu) \ud t + \s \ud B^H_t.
\end{equation}
It is a stationary Gaussian process with mean $E[X_t] = \mu$; one of the very few non-white noise SDEs for which analytical calculations are possible.  In fact, it can be shown~\citep[Remark 2.4]{cheridito-et-al03} that $X_t$ has autocorrelation function
\begin{equation}\label{foucov}
\cov(X_s, X_{s+t}) = \s^2 \Gamma(2H+1)\sin(\pi H) \intii e^{2\pi i t \xi} \frac{\abs{2\pi \xi}^{1-2H}}{\g^2 + (2\pi\xi)^2} \ud \xi
\end{equation}
for any $0 < H < 1$.  While the inverse Fourier transform in~\eqref{foucov} can be approximated numerically, various investigations on our behalf revealed that such a calculation is highly susceptible to roundoff error, and should be carefully monitored (Appendix~\ref{fOU-acfex}).  On the other hand, the complete data likelihood at resolution level $k$ is available directly:
\[
\log(\hat L(\g, \mu, \s, H \| \set X k)) = -\frac 1 2 \left[(\d B^H_{(k)})' V^{-1} (\d B^H_{(k)}) + \log(\abs{V}) + M_k \log(\s^2)\right],
\]
where $\d B^H_{k,n} = \tfrac 1 \s\left\{\d X_{k,n} + \g(X_{k,n}-\mu)\ud t_k\right\}$, and $V$ is a Toeplitz matrix with 
\begin{equation}\label{fgn-acf}
V_{ij} = \frac {(\d t_k)^{2H}} 2 \left(\abs{i-j+1}^{2H} + \abs{i-j-1}^{2H} - 2 \abs{i-j}^{2H}\right).
\end{equation}
Since the fBM residuals $\d B^H_{(k)}$ are linear functions of the complete data, the $\set X k$ themselves are multivariate Normal given $\t = (\g,\mu)$, $\s$, and $H$.  Thus, the missing data $X_\miss$ can be integrated out to yield the marginal Euler-Maruyama posterior $\hat p_k(\t, \s, H \| X_\obs)$ directly (Appendix~\ref{fOU-acfint}).

\subsection{Simulation Experiment}

Using the analytical autocorrelation~\eqref{foucov}, $N+1 = 301$ observations $X_\obs = (\rv [0] X {300})$ were simulated from the fOU process with parameters $\g = 1$, $\mu = 0$, $\s = 1$, $H = .75$, and interobservation time $\d t = 1$.  These data are displayed in Figure~\ref{fou-data}.
\begin{figure}[!h]
	\centering
		\includegraphics[width=1.00\textwidth]{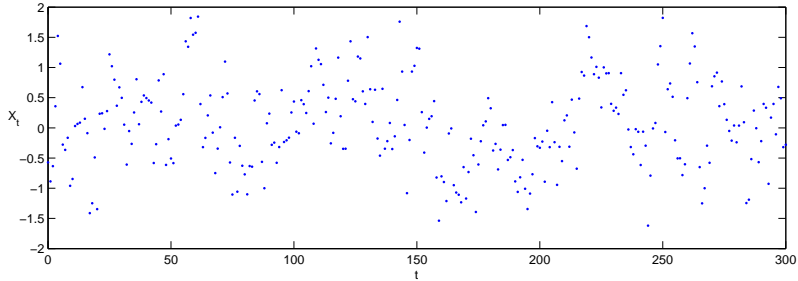}
	\caption{Observations of an fOU process with interobservation time $\d t = 1$ and parameters $\g = 1$, $\mu = 0$, $\s = 1$, and $H = .75$.}
	\label{fou-data}
\end{figure}
Using the standard noninformative prior
\[
\pi(\gamma, \mu,\sigma,H) \propto \g/\s,
\]
posterior inference for the parameters was conducted with Euler-Maruyama approximations at levels $k = 0,1,2,3$.  For these Gaussian models, both $\sigma$ and $\mu$ were integrated out using a method of least-squares.  To illustrate the procedure, note that for $k = 0$, the Euler-Maruyama likelihood function can be framed in a regression context by writing
\begin{equation}\label{fouls}
y_i = \eta \d t + \epsilon_i,
\end{equation}
where $\eta = \gamma\mu$, $y_i = X_{i+1} +(\g \d t - 1)X_i$, and we have correlated errors
\begin{equation}\label{eps}
(\rv [0] \epsilon {N-1}) \sim \N(0, \s^2 V),
\end{equation}
with $V$ given by~\eqref{fgn-acf}.   The noninformative prior becomes $\pi(\gamma,\eta,\s,H)\propto 1/\s$, for which the marginal posterior distribution $\hat p_0(\g, H \| X_\obs)$ can be calculated directly (Appendix~\ref{least-squares}).  

Using this technique, Euler-Maruyama marginal densities of $\g$ and $H$ computed without Monte Carlo error are displayed in Figure~\ref{fou-em-post}.
\begin{figure}[!h]
	\centering
		\includegraphics[width=1.00\textwidth]{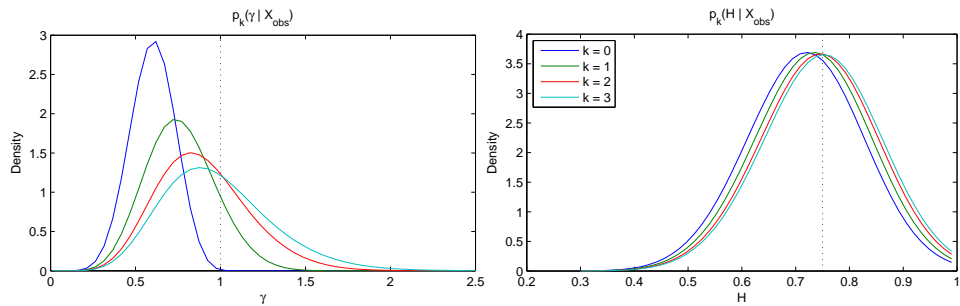}
	\caption{Euler-Maruyama posteriors of $\g$ and $H$.}
	\label{fou-em-post}
\end{figure}
For this particular dataset, the Euler-Maruyama posteriors at level $k=3$ are very close to those of level $k = 2$, suggesting that there is little difference between level $k = 3$ and the true SDE posterior.

To compare with these analytic results, two MCMC samplers were run at resolution level $k = 3$.  The first uses a Gibbs sampling approach, drawing from $\hat p_k(\t,\s,H \| X_\comp)$ using componentwise Metropolis-within-Gibbs, and from $\hat p_k(X_\miss \| X_\obs, \t,\s,H)$ using the HMC algorithm described in Section~\ref{sampling-alg}.  The second MCMC sampler uses HMC to update all the random variables in $p_k(X_\miss, \t,\s,H \| X_\obs)$.  Because the partial derivative $\del H \Omega(X_\miss, \t, \s, H)$ must be computed numerically, the second MCMC sampler takes roughly twice as long per iteration as the first.  In order to establish a basis of comparison, the two MCMC samplers were run for 1,000,000 and 500,000 iterations respectively.

The posterior distributions of $\gamma$ and $H$ for each sampler are displayed alongside the analytic posteriors in Figure~\ref{fou-mcmc-post}.
\begin{figure}[!h]
	\centering
		\includegraphics[width=1.00\textwidth]{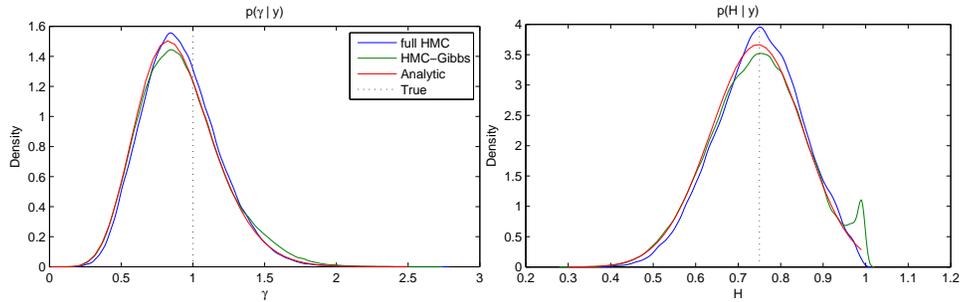}
	\caption{MCMC posteriors of $\g$ and $H$ for $k = 3$.}
	\label{fou-mcmc-post}
\end{figure}
The autocorrelations of these samplers are shown in Figure~\ref{fou-mcmc-acf}, where for comparison, the output of the Gibbs sampler (which had twice as many MCMC iterations) was thinned by a factor of two.
\begin{figure}[!h]
	\centering
		\includegraphics[width=1.00\textwidth]{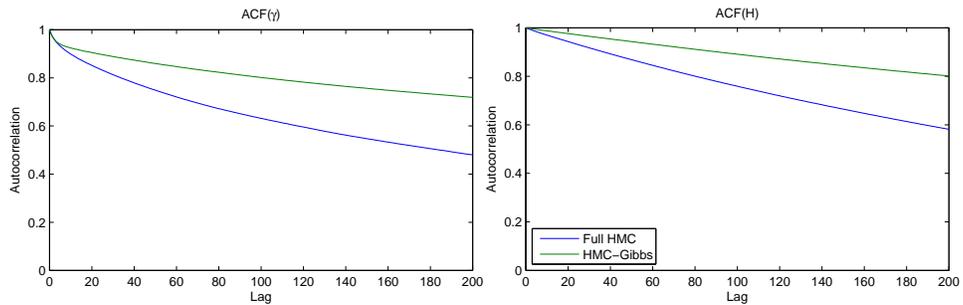}
	\caption{MCMC autocorrelations of $\g$ and $H$ for $k = 3$.}
	\label{fou-mcmc-acf}
\end{figure}
The full HMC sampler is appreciably more efficient that the Gibbs sampler, particularly in its ability to escape a deep local mode arising when $H \approx 1$.  Further evidence of this mode is given by the trace plots of $H$ in Figure~\ref{fou-mcmc-trace}.
\begin{figure}[!h]
	\centering
		\includegraphics[width=1.00\textwidth]{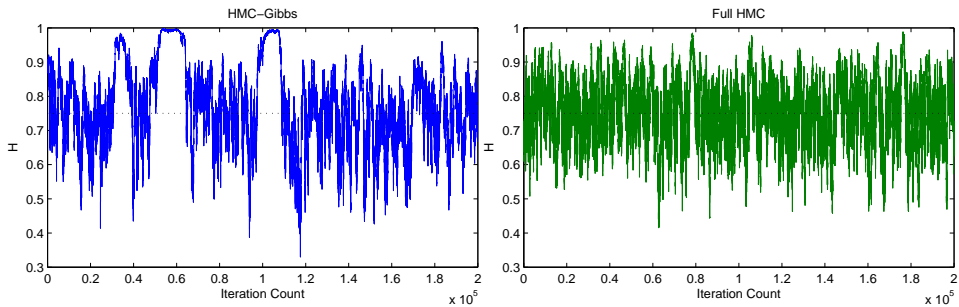}
	\caption{MCMC trace plots of $H$ for $k = 3$.}
	\label{fou-mcmc-trace}
\end{figure}


\section{Application: The Fractional CIR Model for Modeling Short-Term Interest Rates}\label{fCIRex}

The fractional Cox-Ingersol-Ross (fCIR) process is given by the SDE
\begin{equation}\label{fcir}
\ud X_t = -\g(X_t - \mu) \ud t + \s X_t^{1/2} \ud B_t^H.
\end{equation}
The original CIR process~\citep{cox-et-al85} with $H = \tfrac 1 2$ is a popular model for financial assets, admitting a closed-form solution for the transition density as a non-central chi-squared distribution~\citep{kou-et-al12}.  When $H \neq \tfrac 1 2$, this process is no longer analytically tractable.  Some work has been done to extend the CIR process for $H > \tfrac 1 2$ for the special case of $\mu = 0$~\citep{fink-kluppelberg11}.  To our knowledge, we present the first systematic treatment of the fCIR process for unrestricted parametrizations, including any $0 < H < 1$.

The transformation to unit diffusion for the fCIR process is $Y_t = 2X_t^{1/2}$, which yields the SDE
\[
\ud Y_t = (\b/Y_t - \tfrac 1 2 \g Y_t) \ud t + \s \ud B_t^H,
\]
with $\b = 2\g\mu - \tfrac 1 2\s^2$.  We use this model to analyze the long-term memory of 3-Month US Treasury Bills, recorded daily between January 1954 and June 2013 (Figure~\ref{fcir-tbill}a).  This period of 15,508 days was divided into 1000 overlapping segments of $N+1 = 5\times 252 = 1260$ days, or about 5 years (by convention there are 252 trading days per year).  For each of these 1000 datasets we computed a posterior distribution of $H$.
\begin{figure}[!h]
	\centering
		\includegraphics[width=1.00\textwidth]{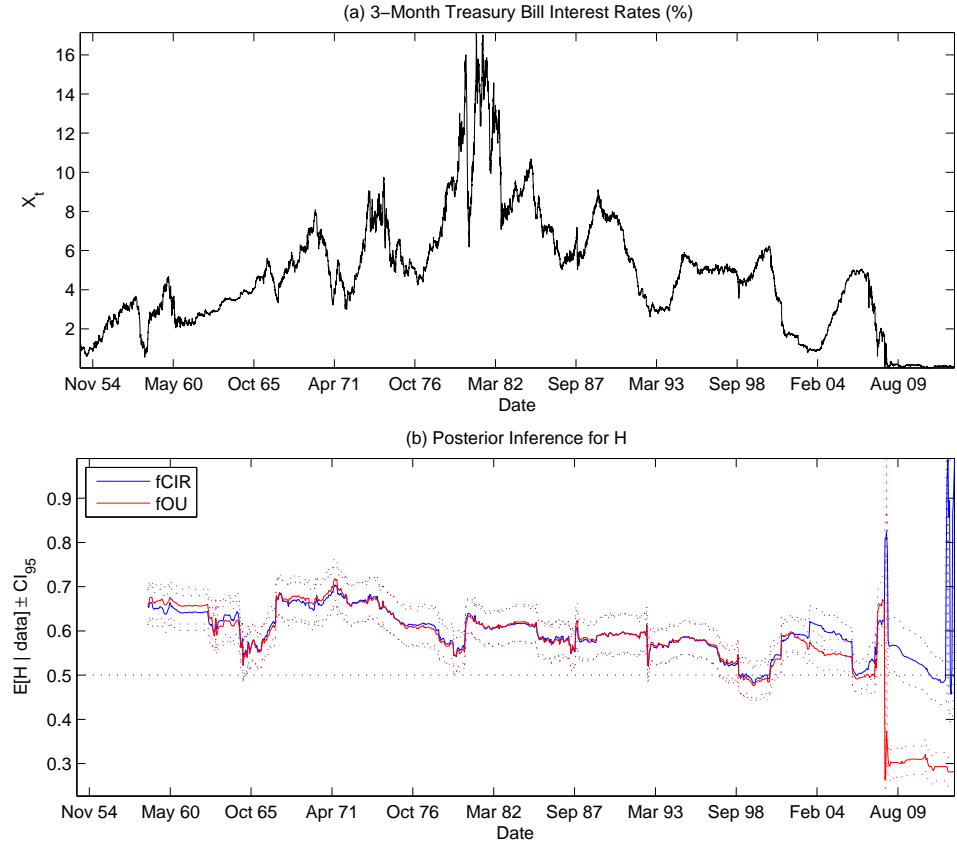}
	\caption{Posterior predictions for the Hurst parameter for 3-month US treasury bills.}
	\label{fcir-tbill}
\end{figure}

Since there is no closed-form likelihood for the parameters of the fCIR process the Euler-Maruyama approximations are used instead.  While these incur a rather substantial computational burden once the missing data is introduced ($k > 0$), preliminary inference without missing data ($k = 0$) can be efficiently accomplished by a variant of the least-squares method described in Section~\ref{fOUex}.  In this case, the regression model becomes
\begin{equation}\label{fcirls}
y_i = -\tfrac 1 2 \g Y_i \d t + \beta \d t/Y_i + \epsilon_i,
\end{equation}
with $y_i = Y_{i+1}-Y_i$, $\d t = 1/252$ (time units of years), and $\epsilon_i$ as in~\eqref{eps}.  For the noninformative prior
\[
\pi(\g,\b,\s,H) \propto 1/\s,
\]
parameters $\g$, $\b$, and $\s$ can be integrated out to yield the marginal posterior $\hat p_0(H \| X_\obs)$ directly.  However, the regression model assumes that $\g, \b \in \mathbb R$, while the fCIR model imposes the restrictions $\g > 0$ and $\b + \tfrac 1 2 \s^2 > 0$ (which corresponds to $\g, \mu > 0$ in the original parametrization).  To comply with these restrictions, the joint posterior of all parameters $\t = (\g, \b)$, $\s$, $H$ is simulated from the unrestricted regression model, using the marginal distribution of $H$ and the analytic conditional distribution of the remaining parameters:
\begin{align*}
\s^2 \| H, X_\obs & \sim \textrm{Inv-Gamma}(a, b) \\
(\g, \b) \| \s^2, H, X_\obs & \sim \N_2(\lambda, \s^2 \Omega),
\end{align*}
where the values of $a$, $b$, $\lambda$, and $\Omega$ are given in Appendix~\ref{least-squares}.  Once a Monte-Carlo draw from the regression model distribution has been generated, it is accepted only if it satisfies the parameter restrictions.  This rejection algorithm provides a quick and simple method for generating samples from the true Euler-Maruyama posterior distribution at level $k = 0$.  For the daily frequency of observations $\d t = 1/252$ considered in this study, a more computationally intensive MCMC analysis revealed that $k = 0$ produced similar inferential results as for higher levels $k \ge 1$.  Thus, the approximate posteriors at $k = 0$ appear to be a very accurate proxy for those of the true fCIR model.

For each of the 1000 subsets of the Treasury Bill Interest Rate data, the posterior mean and 95\% credible intervals of the Euler-Maruyama approximation $\hat p_0(H \| X_\obs)$ are plotted in Figure~\ref{fcir-tbill}b (blue solid and dotted lines).  These are aligned on the x-axis with the last day in the subset which was used to fit the fCIR model.  For comparison, an fOU model is also fit to each of the 1000 datasets, after transforming the interest rates to the log scale (red lines).

Both models have very similar posterior distributions of $H$, generally situated in the range of 0.55-0.65.  Interestingly, the two models starkly diverge in their findings as of the onset of the Global Financial Crisis beginning in 2007.  As of this point, the fCIR model reports a larg positive correlation in the noise, whereas the fOU model picks up a large negative correlation.  This is because the interested rates, which nearly drop to zero after the market crash, appear to fluctuate considerably when taken on the log scale.  These fluctuations translate to negatively correlated noise when fit by the fOU model.  On the other hand, these same near-zero interest rates are almost constant on the regular scale, the exhibiting positive correlation from the fCIR model's perspective.


\section{Discussion}
This article presents a likelihood-based framework for doing parameter inference
for SDEs driven by fractional Brownian motion.  Central to the framework is a simple discretization scheme which carries over much of the statistical insight from the white noise case.  While this paper focuses on SDEs driven by fractional Brownian motion, the methodology can be applied to other types of driving noise as well.  There are several other natural extensions of this work, a few of which are mentioned below.

On the theoretical side, there are two recent papers \cite{hairer-pillai11, hairer-pillai11b} which show the ergodicity of SDEs given driven fBM. Using these results, it might be plausible to show the posterior consistency of the parameters for the above procedure. For continuous data, a lot of details are worked out in \cite{prakasarao11}. However, for discretely observed data, not much is known \cite{saussereau12}.


An important line of research is to find methods for speeding up the computational time by using more sophisticated algorithms for posterior sampling.  In the HMC sampler described in Section~\ref{sampling-alg}, the computational bottleneck is the inner product involved in the calculating the density of the fBM increments.  While the Durbin-Levinson algorithm is $O(2^k)$, a number of ``superfast'' algorithms scaling polynomially in $k$ have recently been proposed~\citep{ammar-gragg88, stewart03, chandrasekeran-et-al07}.  Alternatively, the Girsanov transformation of~\cite{norros99} can be used to construct MCMC algorithms whose mixing time does not deteriorate with $k$.  Using  this, it will be of interest to develop an inference scheme adapting the methods developed in \cite{beskos-et-al06b}.

There are several non-trivial challenges for extending this work to the multivariate case.  While a solution to the non-white noise SDE has been pioneered by the ``rough paths'' approach of T.\ Lyons~\citep{lyons-quian02, papavasiliou-ladroue11}, most multivariate SDEs cannot be transformed to have constant diffusion, and thus higher-order discretization schemes would be required.  While these are commonly used for simulation in the diffusion literature, it is unclear whether these schemes yield a closed-form density which in turn can be used to construct the likelihood.  It should be noted that a non-constant diffusion poses some restriction on the viability of the Girsanov transformation for multivariate SDEs as well.




\section{acknowledgements}
The authors thank Martin Hairer, Samuel Kou, Jonathan Mattingly, Ivan Nourdin, Tessy Papavasiliou, Gareth Roberts, Andrew Stuart, Sami Tindel, Robert Wolpert for many helpful conversations. NSP is partially supported by the NSF grant DMS-1107070.

\appendix



%

\section{Autocorrelation of the fOU Process}\label{fOU-acfex}

In~\citep[Remark 2.4]{cheridito-et-al03}, the autocorrelation of the fOU process is given by the inverse Fourier transform,
\[
\g(t) = \s^2 \Gamma(2H+1)\sin(\pi H)\intii e^{2\pi i t \xi} \frac{\abs{2\pi\xi}^{1-2H}}{\g^2 + (2\pi\xi)^2}.
\]
For $H > \tfrac 1 2$, we have
\[
\F_t^{-1}\left\{\Gamma(2H+1)\sin(\pi H) \abs{2\pi \xi}^{1-2H}\right\} = H(2H-1) \abs{t}^{2H-2}
\]
and
\[
\F_t^{-1}\left\{\frac 1 {\g^2 + (2\pi \xi)^2} \right\} = \frac{e^{-\g\abs{t}}}{2\g},
\]
such that the convolution property of the Fourier transform gives
\[
\g(t) = \frac{\s^2 H(2H-1)}{2\g} \intii e^{-\g\abs{t-u}} \abs{u}^{2H-2} \ud u.
\]
The indefinite integral can be reduced to a definite integral,
\[
\g(t) = \frac{\s^2H(2H-1)}{2\g} \left\{\frac{e^{-\g\abs{t}}\Gamma(2H-1) + e^{\g\abs{t}} \Gamma(2H-1, \abs t)}{\g^{2H-1}} + e^{-\g \abs t} \int_0^t e^{\g u} u^{2H-2} \ud u\right\},
\]
where $\Gamma(\a)$ and $\Gamma(\a, t)$ denote the complete and (upper) incomplete Gamma functions.  We found that this transformation considerably improved numerical stability for $\abs t > 50$ and $H > .6$.

As for $0 < H < \tfrac 1 2$, in this case we have
\[
\F_t^{-1}\left\{\abs{2\pi \xi}^{1-2H} \right\} = \frac{-\Gamma(2-2H)\cos(\pi H)}{\pi} \abs{t}^{2H-2}.
\]
However, the Rieman integral
\[
\intii e^{-\g\abs {t-u}} \abs{u}^{2H-2} \ud u
\]
is infinite because of the behavior as $u \to 0$.  On the other hand, the inverse Fourier transform is usually approximated numerically by its discrete counterpart.  However, we found that for $H < .4$, the roundoff error was still quite significant with $2^{24} \approx 17$ million evaluation points, which, even using the FFT algorithm, was an order of magnitude slower than the our highest order Euler-Maruyama likelihood approximation ($k = 3$).

\section{Approximate Euler Density for the fOU process}\label{fOU-acfint}

Suppose that $X_{(k)} = (\rv [0,k] X {k,N_k})$ are the level $k$ complete data observations of an fOU process,
\[
\ud X_t = -\g(X_t - \mu) \ud t + \s \ud B_t^H.
\]
The Euler-Maruyama complete data density is determined by the recursion
\[
X_{k,n+1} = X_{k,n} - \g(X_{k,n} - \mu) \d t_k + \s \d B^H_{k,n},
\]
where $\rv [k,0] {\d B^H} {k,N_k-1}$ is a mean-zero stationary Gaussian process.  By linearity, it follows that $X_{(k)} = (\rv [k,0] X {k,N_k})$ is also Gaussian.  That is, dropping the subscript $k$, we have
\[
\begin{pmatrix} X_{1} \\ X_{2} \\ \vdots \\ X_{N} \end{pmatrix} = \begin{pmatrix} \vp_0			&  0  & \cdots & 0 \\
																																							\vp_1			& \vp_0   & \ddots & \vdots \\
																																							\vdots 			& \ddots  & \ddots & 0 \\
																																							\vp_{N-1}	& \vp_{N-2}	& \cdots & \vp_0
																															\end{pmatrix} \begin{pmatrix} \s \d B^H_0 \\ \s \d B^H_1 \\ \vdots \\ \s \d B^H_{N-1} \end{pmatrix} + 
																															\begin{pmatrix} \vp_1 \\ \vp_2 \\ \vdots \\ \vp_N \end{pmatrix} X_0 + \g\mu \d t,
\]
where 
$\vp_0 = 1$ and $\vp_n = -\g \d t \vp_{n-1}$.  In matrix form, this equation becomes
\[
X = \s A \d B^H + b,
\]
such that
\[
X \sim \N(b, \s^2 A V A'),
\]
where $V$ is the Toeplitz variance matrix of the fBM increments given in~\eqref{fbmacf}.  The mean and variance matrix in this representation lead to a Gaussian density,
\[
p(X_{k,1}, \ldots, X_{k,N_k} \| X_{k,0}, \g, \mu, \s, H).
\]
The level $k$ approximation to the observed data density,
\[
p(X_1, \ldots, X_n \| X_0, \g, \mu, \s, H),
\]
is also Gaussian and can be obtained by selecting the appropriate elements of $b$ and $\s^2 A V' A$.  Note that for large $N_k = 2^k N$, the matrix multiplication $A V' A$ can be efficiently computed by embedding the Toeplitz matrix $V_{N_k\times N_k}$ with first row $(v_0, \ldots, v_{N_k-1})$ into a circulant matrix $C_{(2N_k-2)\times (2N_k-2)}$, with first row given by
\[
(v_0, \ldots, v_{N_k-1}, v_{N_k-2}, \ldots, v_1).
\]
Inner products involving $C$ can easily be calculated since $C$ is diagonalizable by the discrete Fourier transform matrix $\F$.  By padding $A$ with the appropriate number of zeros,
\[
A_{N_k\times N_k} \to B = \begin{pmatrix} A & 0_{N_k \times (N_k-2)}\end{pmatrix},
\]
we have
\begin{equation}\label{fftquad}
A V A' = B C B' = B_\F D B_\F^\dagger,
\end{equation}
where $D = \F C \F^{-1}$ is a diagonal (complex-valued) matrix, and
\[
B_\F = B \F_{2N-2}^{-1} = (\F_{2N-2} B')^\dagger,
\]
where $\dagger$ denotes the conjugate transpose.  This method of taking inner products produces a considerable acceleration over the direct approach for $N_k > 1000$. 


\section{Marginal and Conditional Posterior Distributions for Regression-Type Models}\label{least-squares}

Suppose that the likelihood function of a given model can be written in a form where the data $Y = (\rv y n)'$ is subject to the linear regression model
\begin{equation}\label{modeluni}
Y \| \b, \s, \t \sim \N_n(X \b, \s^2 V),
\end{equation}
where $X_{n\times d} = X(\t)$ and $V_{n\times n} = V(\t)$ are functions of $\t$.  This is the case for both the fOU and fCIR Euler approximations at $k = 0$, as shown in~\eqref{fouls} and~\eqref{fcirls}.

For given $\t$, consider the block matrix
\[
[Y X]' V^{-1}[Y X] = R_{d+1\times d+1} = \begin{pmatrix} s & U' \\ U & T\end{pmatrix},
\]
where $s_{1\times 1} = s(\t)$, $U_{d\times 1} = U(\t)$, and $T_{d\times d} = T(\t)$ all depend on $\t$.  Using these quantities, the log-likelihood function for the model in~\eqref{modeluni} can be written as
\[
\begin{split}
l(\b, \s, \t \| Y) & = -\frac 1 2 \left(\frac{(Y-X\b)'\inv V (Y-X\b)}{
\s^2} + n \log(\s^2) + \log(\abs{V}) \right) \\
& = - \frac 1 2 \left(\frac{(\b -\hat\b)' T (\b - \hat\b) + S}{\s^2} + n \log(\s^2) + \log(\abs{V}) \right),
\end{split}
\]
where $\hat \b = \inv T U$ and $S = s - U' \hat \b$.  The conjugate prior $\pi(\b,\s,\t)$ for this model is of the form
\begin{equation}\label{conjprior}
\begin{split}
\t & \sim \pi(\t), \\
\s^2 \| \t & \sim \textrm{Inv-Gamma}(\a, \g) \propto (\s^2)^{-\a-1} \exp(-\g/\s^2), \\
\b \| \s, \t & \sim \N_d(\l, \s^2\inv \Omega).
\end{split}
\end{equation}
This results in the posterior distribution
\begin{align*}
\t \| Y & \sim \frac{\pi(\t)} {\left((\hat \g + \g)^{2\a+n}\abs{T+\Omega}\abs{V}\right)^{1/2}} \\
\s^2 \| \t, Y & \sim \textrm{Inv-Gamma}\left(\a + n/2, \g + \hat \g \right) \\
\b \| \s, \t, Y & \sim \N_d\left(\hat \l, \s^2(T + \Omega)^{-1}\right),
\end{align*}
where
\begin{align*}
\hat \l & = (T+\Omega)^{-1}(U + \Omega \l), \\
\hat \g & = \tfrac 1 2 \big(s + \l'\Omega\l - \hat \l'(T+\Omega)\hat \l\big).
\end{align*}
Note that the popular noninformative prior
\[
\pi(\b, \s^2, \t) \propto \pi(\t)/\s^2
\]
is also a member of the conjugate family~\eqref{conjprior}.  For this noninformative prior, the posterior parameter distribution simplifies to
\begin{align*}
\t \| Y & \sim \frac{\pi(\t)}{(\hat \g^{n-d} \abs{T}\abs{V})^{1/2}} \\
\s^2 \| \t, Y & \sim \textrm{Inv-Gamma}((n-d)/2, \hat \g) \\
\b \| \s, \t, Y & \sim \N_d(\hat \b, \s^2T^{-1}),
\end{align*}
where $\hat \g = \tfrac 1 2 s$.


\bibliographystyle{numref}
\bibliography{fbm-ref}

\end{document}